 \documentclass[smallabstract,smallcaptions]{dccpaper}

\usepackage{epsfig}
\usepackage{citesort}
\usepackage{amssymb}
\usepackage{color}
\usepackage{url}

\usepackage{amsthm}
\usepackage{amsmath}
\usepackage{amssymb}
\usepackage{color}
\usepackage{url}
\usepackage{epsfig}
\usepackage{citesort}
\usepackage{multirow}
\usepackage{xcolor}
\usepackage{comment}
\usepackage{thm-restate}

\usepackage{appendix}
\usepackage{graphicx}
\usepackage{enumerate}
\usepackage{tikz}
\usepackage[linesnumbered,ruled,vlined]{algorithm2e}
\usepackage{algpseudocode}
\newtheorem{theorem}{Theorem}
\newtheorem{observation}[theorem]{Observation}
\newtheorem{corollary}[theorem]{Corollary}
\newtheorem{proposition}[theorem]{Proposition}

\newtheorem{remark}[theorem]{Remark}
\newtheorem{lemma}[theorem]{Lemma}

\graphicspath{ {./Figures/} }

\newcommand{\degr}{\textsf{deg}}
\newcommand{\neighbor}{\textsf{neighbor}}
\newcommand{\adj}{\textsf{adj}}
\newcommand{\rank}{\textsf{rank}}
\newcommand{\select}{\textsf{select}}

\newlength{\figurewidth}
\newlength{\smallfigurewidth}

\begin{document}

\title
{\large
\textbf{Succinct Data Structure for Graphs with $d$-Dimensional $t$-Representation}
}

\author{%
Girish Balakrishnan$^{\ast}$, Sankardeep Chakraborty$^{\dag}$, Seungbum Jo$^{\ddagger}$, \\ N S Narayanaswamy$^{\ast}$, and Kunihiko Sadakane$^{\dagger}$\\[0.5em]
{\small
\begin{minipage}
{\linewidth}
\begin{center}
\begin{tabular}{ccc}
$^{\ast}$Indian Institute of Technology Madras, & \hspace*{0.5in} & $^{\dagger}$University of Tokyo,\\
Chennai, India && Tokyo, Japan \\
\url{girishb@cse.iitm.ac.in} && \url{sankardeep.chakraborty@gmail.com}\\
\url{swamy@cse.iitm.ac.in} && \url{sada@mist.i.u-tokyo.ac.jp}\\
\end{tabular}
\begin{tabular}{ccc}
& $^{\ddagger}$Chungnam National University,     &  \\
&   Daejeon, South Korea  & \\
& \url{sbjo@cnu.ac.kr} \\
\end{tabular}
\end{center}\end{minipage}}
}

\maketitle
\begin{abstract}
Erd\H{o}s and West (Discrete Mathematics'85) considered the class of $n$ vertex intersection graphs which have a {\em $d$-dimensional} {\em $t$-representation}, that is, each vertex of a graph in the class has an associated set consisting of at most  $t$ $d$-dimensional axis-parallel boxes.  In particular, for a graph $G$ and for each $d \geq 1$, they consider $i_d(G)$ to be the minimum $t$ for which $G$ has such a representation.  For fixed $t$ and $d$, they consider the class of $n$ vertex labeled graphs for which $i_d(G) \leq t$, and prove an upper bound of $(2nt+\frac{1}{2})d \log n - (n - \frac{1}{2})d \log(4\pi t)$ on the logarithm of size of the class. 

In this work, for  fixed $t$ and $d$ we consider the class of $n$ vertex unlabeled graphs which have a {\em $d$-dimensional $t$-representation}, denoted by $\mathcal{G}_{t,d}$.  We address the problem of designing a succinct data structure for the class $\mathcal{G}_{t,d}$ in an attempt to generalize the relatively recent results on succinct data structures for interval graphs (Algorithmica'21). To this end, for each $n$ such that $td^2$ is in $o(n / \log n)$, we first prove a lower bound of $(2dt-1)n \log n - O(ndt \log \log n)$-bits on the size of any data structure for encoding an arbitrary graph that belongs to $\mathcal{G}_{t,d}$. 

We then present a $((2dt-1)n \log n + dt\log t + o(ndt \log  n))$-bit data structure for $\mathcal{G}_{t,d}$ that supports navigational queries efficiently. Contrasting this data structure with our lower bound argument, we show that for each fixed $t$ and $d$, and for all $n \geq 0$ when $td^2$ is in $o(n/\log n)$  our data structure for $\mathcal{G}_{t,d}$ is succinct.

As a byproduct, we also obtain succinct data structures for graphs of bounded boxicity (denoted by $d$ and $t = 1$) and graphs of bounded interval number (denoted by $t$ and $d=1$) when $td^2$ is in $o(n/\log n)$. 

\end{abstract}
\section{Introduction}
Research in succinct data structures has been a classical problem. Representations for a class of graphs with $n$ vertices that use space equal to information-theoretic lower bounds up to lower order term are called \textit{succinct representations}. Such representations for a graph class $\mathcal{G}$ are obtained by first proving a lower bound on the size of the graph class, that is $|\mathcal{G}| \ge N$, followed by designing a $\log N + o(\log N)$-bit data structure for graphs in $\mathcal{G}$. Additionally, a key challenge is to consider whether there is a succinct representation for $G=(V,E)$ in $\mathcal{G}$ that supports the following basic and fundamental navigational queries for each pair of vertices $u, v \in V$:
\begin{itemize}
\itemsep0em
\item $\adj{}(u, v)$: returns "YES" if and only if vertices $u$ and $v$ are adjacent in $G$.
\item $\neighbor{}(u)$: returns all the vertices in $V$ that are adjacent to vertex $u$.
\item $\degr{}(u)$: returns the number of vertices adjacent to vertex $u$.
\end{itemize}
In this paper, we present succinct representations for a class of $n-$vertex graphs with $d-$dimensional $t-$interval representation. The earliest work in the design of succinct representation for graph classes is by Itai and Rodeh~\cite{IR}, in which they gave a $({\frac{3}{2}}n \log n + O(n))$-bit representation for the class of $n$ vertex labeled planar graphs and also showed a $n \log n + O(n)$ information-theoretic lower bound. It was in the work of Jacobson~\cite{Jacobson1989} about three and a half decades ago, on the succinct representation for class of static unlabeled trees and planar graphs, that the efficiency of time along with space was considered for the first time. Since then, extensive research within this realm has yielded a plethora of such data structures (as comprehensively introduced in~\cite{gonzalo}) for a diverse array of combinatorial entities including but not limited to trees~\cite{NavarroS14}, arbitrary graphs~\cite{FarzanM13}, planar maps~\cite{AleardiDS08}, finite automaton~\cite{ChakrabortyGSS23}, functions~\cite{MunroR04}, permutations~\cite{MunroRRR03}, posets~\cite{MunroN16}, bounded treewidth graphs~\cite{FarzanK11}, texts~\cite{GagieNP20}, sequences~\cite{NavarroN14}, and countless others. For arbitrary graphs with $n$ vertices and $m$ edges, Farzan and Munro~\cite{FarzanM13} have shown that  it is possible to  obtain space optimal representation within a $(1+\epsilon)$ multiplicative factor of information-theoretic lower bound for any constant $\epsilon > 0$. Only for sparse graphs, that is, $m=o(n^{\delta})$ for any constant $\delta>0$, a succinct representation is obtained in~\cite{FarzanM13}. As it stands today, this field has attained a state of maturity, showcasing its profound evolution. 
%Among these, we are mainly concerned with succinct data structures for graphs. 
%unless that $m$ is $o(n^{\delta})$ or $\omega(n^{2-\delta})$, for any  $\delta>0$. 

One of the most well-studied graph classes is the class of intersection graphs; see McKee and McMorris~\cite{MM} for more details on intersection graphs. Succinct  data structures for intersection graphs, particularly interval graphs, and their generalizations have already appeared; see Golumbic~\cite{agtpg} for more on interval graphs. For interval graphs, Acan et al. in~\cite{HSSS} (and later He et al.~\cite{he}) give a $(n \log n + O(n))$-bit succinct data structure that supports adjacency, neighbourhood and degree queries in constant time. Class of chordal graphs is a strict super-class of interval graphs and Munro and Wu in~\cite{Munro_Wu} have given a $(n^2/4 + o(n^2))$-bit succinct data structure that supports the queries efficiently. An $(n \log n + o(n  \log n))$-bit succinct data structure for path graphs, a strict super-class of interval graphs and a strict sub-class of chordal graphs, is given by Balakrishnan et al. in~\cite{GSNS};  see~\cite{agtpg} for more details on chordal graphs, path graphs  and interval graphs. A very recent paper by Acan et al.~\cite{HSSKKS2020} gives succinct data structures for families of intersection graphs of generalized polygons on a circle. 

In this paper, we add to this growing body of work by presenting a data structure for the class of $n$ vertex graphs, introduced by Erd\H{o}s and West in~\cite{EW}, with $d-$dimensional $t-$representation, denoted by $\mathcal{G}_{t,d}$. They proved an upper bound of $(2nt+\frac{1}{2})d \log n - (n - \frac{1}{2})d \log(4\pi t)$ on $\log|\mathcal{G}_{t,d}|$ for fixed $t$ and $d$. The class of $d-$boxicity graphs or bounded boxicity, denoted $\mathcal{G}_d$ and class of $t-$interval graphs or bounded interval number, denoted $\mathcal{G}_t$, are obtained when $t=1$ and $d=1$, respectively.  

For $t=d=1$, $\mathcal{G}_{t,d}$ is the class of interval graphs for which Acan et al. have given a succinct representation in~\cite{HSSS}. There seems to be no other result on $\mathcal{G}_{t,d}$ after~\cite{EW}, to the best of our knowledge. Though, $\mathcal{G}_d$ and $\mathcal{G}_t$ are well-studied, the upper bound in~\cite{EW} seems to be the only attempt at bounding the logarithm of their sizes. Further, when it comes to the representation of these graphs, to the best of our knowledge, Spinrad in~\cite{Spinrad95} and very recently Cotumaccio et al.~\cite{prezza23} mentioned that $G \in \mathcal{G}_d$ can be stored using $O(d \log n)$ bits per vertex by storing the coordinates of the boxes corresponding to each vertex. However, we show that this naive representation is not space-wise optimal.

In this paper, apart from proving the enumerative lower bound, we show a matching upper bound by designing a data structure using $((2dt-1)n \log n + 2dtn \log t + o(dtn\log n))$ bits, which is asymptotically equal to that given in Lemma 1 of~\cite{EW} after accounting for the additional $n \log n$ bits used for labeling. There are no known succinct data structures for $\mathcal{G}_{t,d}$ and we address this by proving  that the above mentioned $((2dt-1)n \log n + 2dtn \log t + o(dtn\log n))$-bit data structure  is succinct when $td^2$  is in $o(n/\log n)$. Further, as a byproduct of the succinct data structure for $\mathcal{G}_{t,d}$, we obtain succinct data structures for $t-$interval and $d-$boxicity graphs when $td^2$ is in $o(n/\log n)$. Particularly, for $t-$interval graphs, $\mathcal{G}_t$, which  is a generalization of interval graphs that has received a lot of attention,  we present a lower bound on $\log |\mathcal{G}_t|$, a matching succinct data structure that supports queries efficiently and a conditional hardness result on the time complexity of neighbourhood query for interval number in $\Theta(n)$. 

%For $k=2$, we make the following observation regarding the $2-$interval and $2-$boxicity graph classes. 
%The class of trees is a subset of $\mathcal{G}_{d=2} \cap \mathcal{G}_{t=2}$; see~\cite{TH} and~\cite{Chintan2008} for proofs that show that trees have interval number and boxicity at  most two. However, complete bipartite graphs $K_{m,n}$ have boxicity two but interval number of $\lceil (mn +1)/(m+n) \rceil$; see~\cite{SN} that proves that complete bipartite graphs have boxicity two and~\cite{TH} that proves the expression for interval number of complete bipartite graphs. Further, a complement of a complete matching on $n$ vertices has interval number two but boxicity $n/2$; see~\cite{SW}.  Thus, $\mathcal{G}_{d=2} \cap \mathcal{G}_{t=2} \ne \phi$ but one is not contained in the other. 

%\noindent
\section{Our Main Results} 
We start with a counting lower bound result. In particular, using the method similar to that of partial coloring as demonstrated by Acan et al. in~\cite{HSSKKS2020}, we prove a lower bound of $(2dt-1)n \log n - [4 \log d + 2n \log t + 2n \log \log n + n ] dt - 2n -O(\log n)$  for $\log |\mathcal{G}_{t,d}|$. In order to obtain the lower bound, we first construct a class of graphs $\mathcal{G}'_{t,d}$ from graphs  in $\mathcal{G}_{t,d}$ where some vertices of the graphs belonging to $\mathcal{G}'_{t,d}$ are specially labeled.
%The lower bound for $\mathcal{G}_{t,d}$ is obtained by first constructing a class $\mathcal{G}'_{t,d}$ using a labeling scheme which is called partial coloring in~\cite{HSSKKS2020}. 
Since graphs in $\mathcal{G}'_{t,d}$ are created from graphs  in $\mathcal{G}_{t,d}$ and are partially labeled we have $\log |\mathcal{G}'_{t,d}| \le \log|\mathcal{G}_{t,d}| + \mathcal{E}$ where $\mathcal{E}$ is the logarithm of the total number of graphs created from graphs in $\mathcal{G}_{t,d}$; the inequality is due to over-counting of isomorphic graphs. Next, we construct a proper subclass of $\mathcal{G}'_{t,d}$, denoted $\mathcal{G}^c_{t,d}$ such that we have an exact count for $|\mathcal{G}^c_{t,d}|$. Thus, we establish a relation $\log|\mathcal{G}_{t,d}| + \mathcal{E} \ge \log |\mathcal{G}'_{t,d}| \ge \log |\mathcal{G}^c_{t,d}|$ or $\log|\mathcal{G}_{t,d}| \ge \log |\mathcal{G}^c_{t,d}| - \mathcal{E}$. Using this idea, we have the following theorem:
\begin{restatable}[]{theorem}{lowerbound}
\label{thm:succreplbkinterval}
For $t,d \ge 1$ and $td^2$ in $o(n /\log n)$, $\log|\mathcal{G}_{t,d}|  \ge  (2dt-1) n \log n - [4n \log d + 4n \log t + 2n \log\log n + n ]dt - 2tn -O(\log n)$.    
\end{restatable}

%\noindent
Then, the lower  bounds  for $\mathcal{G}_t$ and $\mathcal{G}_d$ are obtained from Theorem~\ref{thm:succreplbkinterval} as a corollary.

\begin{restatable}[]{corollary}{lbcorollary}
\label{cor:lbtintdbox}  
For $t,d \ge 1$ and $td^2$ in $o(n /\log n)$, we have the following:
\begin{itemize}
    \itemsep0em
    \item $\log |\mathcal{G}_t| \ge (2t-1)n\log n - tn[4\log t + 2 \log\log n + 3] -  O(\log n)$, and 
    \item $\log |\mathcal{G}_d| \ge (2d-1)n\log n - dn[4\log d + 2\log \log n +1] - 2n - O(\log n)$.
\end{itemize}
\end{restatable}

\vspace{2mm}
%\noindent
Next, we proceed to design upper-bound data structures. More specifically, the following theorem gives a data structure for encoding the graphs belonging to $\mathcal{G}_{t,d}$.

\begin{restatable}[]{theorem}{datastructure}
\label{thm:tdintervalub}
Given a graph $G = (V, E)$ of $n$ vertices with $d-$dimensional $t-$representation, there exists $((2dt-1)n \log n + 2dtn \log t + o(dtn\log n))$-bit data structure that can answer $\adj{}(u, v)$ query in $O(dt^2)$ time, and $\neighbor{}(u)$ query in $O(dtn)$ time, for any two vertices $u, v \in V$. Furthermore, when $d=1$, we can answer $\neighbor{}(u)$ query in $O(t^2 (f(n) + \log \log t) \cdot \degr{}(u))$ time. Here, $f(n)$ is any increasing function in $o(\log n)$.
\end{restatable}

\vspace{2mm}
%\noindent
%The data structure of Theorem~\ref{thm:tdintervalub} is constructed using the following two types of succinct data structures:
%\begin{enumerate}
%    \itemsep0em
%    \item {\bf A \rank{}/\select{} data  structure $S_j, 1 \le j \le d,$ on alphabet $\{0,\ldots,t\}$:} For each vertex $v \in V(G)$ and $1 \le p \le t,$ the $p-$th $d-$dimensional box corresponding to $v$, denoted $I_p^v$, has a projection $I_p^v(j)$  on axis $1 \le j \le d$. For $1 \le i \le 2tn, S_j[i]$ stores $2(p-1)$ if $I_p^v(j)$ starts at $i$ and $(2p-1)$ if $I_i^v(j)$ ends at $i$.
%    \item {\bf Permutations:} There are two permutations, namely, $\pi_{(p,j)}(v)$ and $\rho_{(p,j)}(v)$. $\pi_{(p,j)}(v)$ stores the number of $I_p^{v'}(j)$ for some $v' \in V(G)$ that has started at or before $l_{p,j}^v$ where $l_{p,j}^v$ is the left end point of $I_p^v(j)$. Similarly, $\rho_{(p,j)}(v)$ stores the number of $I_p^{v'}(j)$ for some $v' \in V(G)$ that has ended at or before $r_{p,j}^v$ where $r_{p,j}^v$ is the right end point of $I_p^v(j)$.
%\end{enumerate}

%\noindent
Finally, the following theorem shows that the data structure of Theorem~\ref{thm:tdintervalub} is succinct.
\begin{restatable}[]{theorem}{succinct}
\label{thm:dssucc}
For $t,d \ge 1$ and $td^2$ in $o(n/\log n)$, the $((2dt-1)n \log n + 2dtn \log t + o(dtn\log n))$-bit data structure for class of graphs with $d-$dimensional $t-$representation is succinct.
\end{restatable}

\vspace{2mm}
%\noindent
As a corollary of Theorem~\ref{thm:tdintervalub} and other results, we also obtain space-efficient data structures for graphs with bounded edges and degrees. Such graphs are very useful in practice and appear in a variety of applications~\cite{AdlerH18}. 

\noindent
The representation of $\mathcal{G}_{t,d}$ that we give in this paper is called the $t,d-$intersection representation. For $t-$interval graphs, in addition to the succinct data structure that supports navigational queries efficiently, we have the following conditional hardness result that shows that the combinatorial Boolean matrix multiplication (BMM) conjecture, put forth by Henzinger et al. in~\cite{DBLP:conf/stoc/HenzingerKNS15}, establishes a lower bound of $O(n^{2-\epsilon})$ on the neighbourhood query of $t-$interval graphs when presented as a $t,1-$intersection representation.
\begin{restatable}[]{theorem}{hardness}
\label{thm:hardness}
For a given $t, 1$-intersection representation of a graph $G = (V, E)$ with $n$ vertices, if there exists a data structure with construction time $c(t, n)$ that can answer $\neighbor{}(v)$ queries in $\alpha(t, n)$ time for any $v \in V$ then it is possible to devise an algorithm that solves Boolean matrix multiplication of size $n \times t$ and $t \times n$ in the Boolean semi-ring, running in $O(n\alpha(t, n) + c(t, n) + nt)$ time.   
\end{restatable}
\vspace{2mm}
Thus, assuming the validity of the BMM conjecture, the data structure presented in Theorem~\ref{thm:tdintervalub} offers an asymptotically optimal query time (within a polylogarithmic factor) for both $\adj{}$ (in amortized) and $\neighbor{}$ queries when $d = 1$ and $t = \Theta(n)$, given that $G$ is provided as an $t$-interval representation.

\noindent
\textbf{Organization.} The rest of the paper is arranged as follows. Section~\ref{sec:prelims} contains all the preliminary concepts and definitions required for the rest of the paper. Section~\ref{sec:kintlowerbound} gives the lower bound for the size of the class of graphs with $d-$dimensional $t-$representation and as corollary, lower bounds for $d-$boxicity and $t-$interval graphs. Section~\ref{sec:succrep} explains succinct data structures for graphs with $d-$dimensional $t-$representation along with details of the implementation of adjacency query and an efficient neighbourhood query for $\mathcal{G}_t$. Also, for $\mathcal{G}_t$, a conditional hardness proof for neighbourhood query is given. Finally, we conclude in Section~\ref{sec:conclusion} with some open problems.

\section{Preliminaries} \label{sec:prelims}
In this section, we provide definitions for two parameters related to intersection graphs, namely the \textit{interval number} and the \textit{boxicity}, and is provided only for the sake of a complete presentation of the graph classes considered in this paper.  The results in the following sections can be read independently of this presentation. 

A graph $G=(V,E)$ is  an \textit{intersection graph} if for some set  $X$ and $\mathcal{F}=\{X_1,\ldots,X_n\},$\\$ X_i \subset X, 1 \le i \le n$, there exists a bijection $f:V \rightarrow \mathcal{F}$, and $u,v \in V$ are  adjacent if $f(u) \cap f(v) \ne \phi$.

\noindent
{\bf Interval Number.} Let $\mathcal{I}=\{\mathcal{I}_1,\ldots,\mathcal{I}_n\}$ where for $1 \le i \le n, \mathcal{I}_i=\{I^i_1,\ldots,I^i_t\}$ and $I^i_j=[l_j,r_j], 1 \le j \le t,$ is a closed interval in the real line. A graph $G \in \mathcal{G}_t$, if there exists a bijection $g:V \rightarrow \mathcal{I}$ and $u,v \in V$ are adjacent if there exists $I \in g(u)$ and $I' \in g(v)$ such that $I \cap I' \ne \phi$. We say  that $G$ has a  \textit{$t-$interval representation} and the minimum such $t$ is called the \textit{interval number} of $G$. Graphs with interval number $t$ was introduced by Trotter and Harary in~\cite{TH} and studied extensively in ~\cite{EW,WestS84}. Figure~\ref{fig:2-intervalgraph} shows an example graph with interval number two.

\begin{figure}[ht]
\centering
\includegraphics[width=14cm,height=3.5cm]{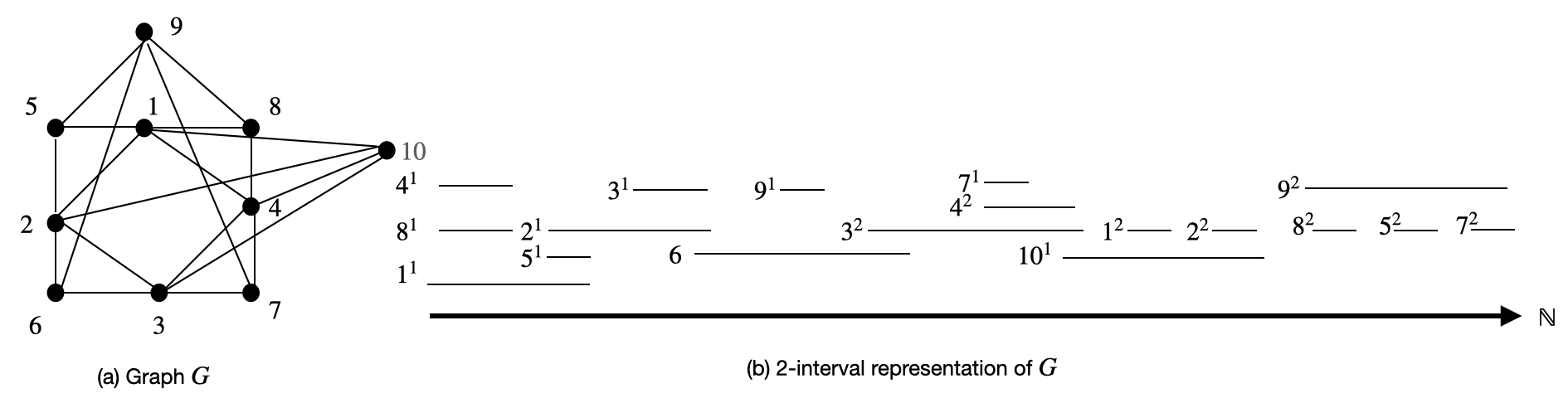}
\caption{$H \in \mathcal{G}_t, t=2$ and its 2-interval representation in (a) and (b), respectively. Observe that $H$ does not have a 2-box representation though $H$ and graph $G$, of Figure~\ref{fig:2-boxicitygraph}, differ in only one vertex 10.}
\label{fig:2-intervalgraph}
\end{figure}

\noindent
{\bf Boxicity.} Consider graph $G=(V,E)$ such  that for every $v  \in V$ there exists a $d-$dimensional axis  parallel box $b(v)=I^v_1 \times \ldots \times I^v_d$,  where $I^v_j=[l_j,r_j], 1 \le j \le d$, is a closed interval on real line of axis $j$ and $\{u,v\} \in E$ if and only if $b(u) \cap b(v) \ne \phi$. We say that $G$ has a \textit{$d-$dimensional box representation} or a $d-$box representation and the minimum such $d$ is called the \textit{boxicity} of $G$. The class of $d-$boxicity graphs, denoted $\mathcal{G}_d$, is the class of graphs that have a $d-$dimensional box representation. Graphs with boxicity $d$ were introduced by Roberts in~\cite{ROBERTS1969} and studied in~\cite{AJS,ADS,ASN,SN,ESPERET2009,SMN}. Figure~\ref{fig:2-boxicitygraph} shows an example graph with boxicity two. 

\begin{figure}[ht]
\centering
\includegraphics[width=8cm,height=4.5cm]{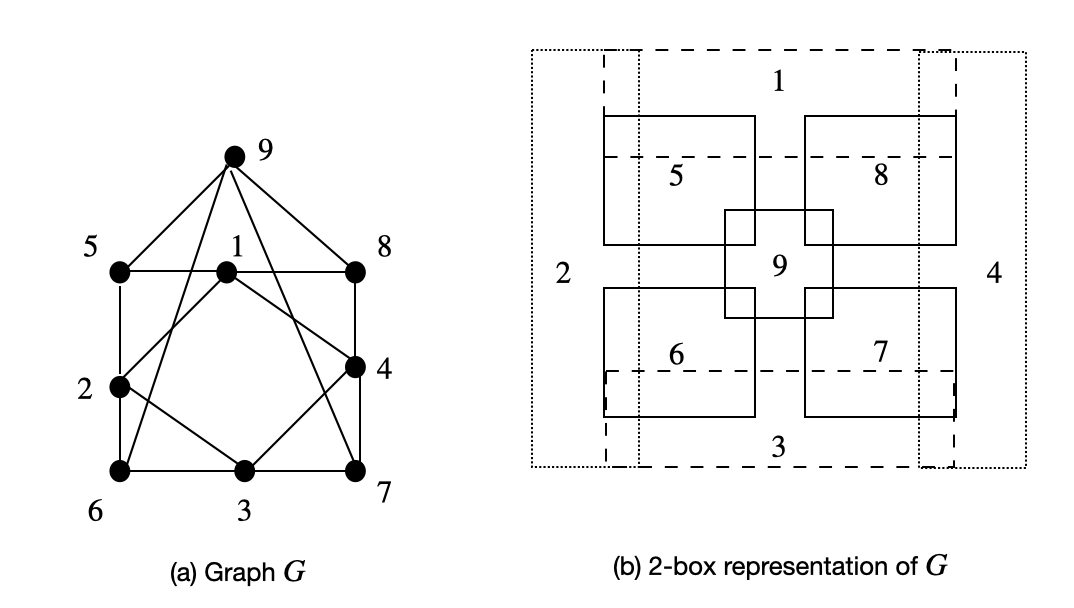}
\caption{$G \in \mathcal{G}_d, d=2$ and its 2-box representation in (a) and (b), respectively.}
\label{fig:2-boxicitygraph}
\end{figure}

\noindent
{\bf Generalization.} Extending the previous two notions, a graph $G$ is said to have a $d-$dimensional $t-$representation if and only if it can be represented as an intersection graph, where each vertex $v$ is associated with a set $I^v$ consisting of $t' \leq t$ disjoint $d$-dimensional intervals $\{I^v_1, \dots, I^{v}_{t'}\}$, where for $1 \le p \le t'$, $I^v_p = [l^v_{p, 1}, r^v_{p, 1}] \times [l^v_{p, 2}, r^v_{p, 2}] \times \dots \times [l^v_{p, d}, r^v_{p, d}]$. From this definition, interval graphs are a special case where both the interval number and boxicity are equal to $1$. We say that $\mathcal{I} = \{I^v \mid v \in V\}$ is the \textit{$d-$dimensional $t$-representation} or in short $t,d-$ intersection representation of $G$, which is a collection of sets corresponding to the vertices in $V$. Also, it is known that determining whether a graph $G$ has a given interval number $t$ or boxicity $d$ is an NP-hard problem for any fixed $t \geq 2$ or $d \geq 2$~\cite{Kratochvil94,WestS84}. 

\noindent
{\em Convention.} In the rest of this paper, we use $[n]$ to denote a set of positive integers $\{1, \dots, n\}$. For a graph class $\mathcal{G}$, the number of graphs in it is denoted by $|\mathcal{G}|$. The terms $d-$dimensional box and $d-$dimensional interval are used interchangeably in the paper. 

\noindent
{\bf Method of Partial Coloring.} Let $\mathcal{G}$ be a graph class. A \textit{partial coloring} of $G$ is the triple $\langle G, U, g \rangle$ where,
\begin{itemize}
    \itemsep0em
    \item $U \subseteq V(G)$ such that for $s>0, |U|=s$, and
    \item $g:U \rightarrow\{1,\ldots,s\}$ is a bijection.
\end{itemize}

\noindent
Vertices in $U$ are said to be colored using colors $\{1,\ldots,s\}$. 
%Two graphs $G_1$ and $G_2$ are isomorphic when there exists a bijection $\Phi: V(G_1) \rightarrow V(G_2)$ such that for every $\{u,v\} \in E(G_1)$ there exists $\{\Phi(u),\Phi(v)\} \in E(G_2)$. 
Two partially colored graphs $\langle H_1,U_1,g_1 \rangle$ and $\langle H_2,U_2,g_2 \rangle$ are said to be different when either:
\begin{enumerate}
    \itemsep0em
    \item $E(H_1) \ne E(H_2)$, or
    \item $E(H_1)=E(H_2)=E(H)$ for some $H \in \mathcal{G}$ and there exists $u \in V(H)\backslash U$ such that its colored neighbourhood in $\langle H_1,U_1,g_1 \rangle$ and $\langle H_2,U_2,g_2 \rangle$ are different.
\end{enumerate}
Else, they are same. The method of counting by partial coloring as given in Theorem 1 of Acan et al.~\cite{HSSKKS2020} can be defined using the following proposition.
\begin{proposition}
\label{prop:partialcoloring}
    Let $\mathcal{G}'$ be the class of partially colored graphs obtained from class of graphs $\mathcal{G}$ by selecting $m$ vertices out of $n$ and coloring them using $m$ distinct colors. Then, $|\mathcal{G}'| \le {n \choose m} m! |\mathcal{G}|$. If there exists a graph class $\mathcal{G}^c \subset \mathcal{G}'$ then $|\mathcal{G}'| \ge |\mathcal{G}^c|$ and $|\mathcal{G}| \ge \frac{|\mathcal{G}^c|}{{n \choose m} m!}$.
\end{proposition}

\noindent
From the definition of partial coloring as a triple $\langle G, U, g \rangle$, we have the following observations regarding $|\mathcal{G}'| \le {n \choose m} m! |\mathcal{G}|$  in Proposition~\ref{prop:partialcoloring}:
\begin{enumerate}
    \itemsep0em
    \item The total number of ways $G$ can be obtained is $|\mathcal{G}|$.
    \item $U$ can be obtained in ${n \choose m}$ ways.
    \item The total number of bijections $g$ is $m!$, that is, the total number of ways $U$ can be distinctly colored using $m$ colors.
    \item Since the same graph can be counted more than once we have the inequality. For instance, consider a complete $n-$vertex graph $G$ in which $m$ vertices are selected and colored distinctly using colors $\{1,\ldots,m\}$. Any permutation of colors among the $s$ vertices results in the same neighbourhood for the uncolored vertices.
\end{enumerate}

\noindent
{\it Remark:} While computing $|\mathcal{G}'|$, indistinguishable partially colored graphs can also be counted since we only require an upper bound, however, this is not the case while computing $|\mathcal{G}^c|$.

\section{Lower Bound for Graphs with $d-$Dimensional $t-$Representation}
\label{sec:kintlowerbound}
The lower bound for the size of the class of graphs with $d-$dimensional $t-$representation, $\mathcal{G}_{t,d}$, is obtained  by implementing Proposition~\ref{prop:partialcoloring}. In order to do this we define two graph classes, namely, $\mathcal{G}'_{t,d}$ and $\mathcal{G}^c_{t,d}$, as follows. Note that $\mathcal{G}'_{t,d}$ corresponds to $\mathcal{G}'$ and $\mathcal{G}^c_{t,d}$ to $\mathcal{G}^c$ of Proposition~\ref{prop:partialcoloring}.

\noindent
{\bf Graph Class $\mathcal{G}'_{t,d}$.} To obtain  a lower bound for $|\mathcal{G}_{t,d}|$ we first consider the class of partially colored graphs with $d-$dimensional $t-$representation, denoted $\mathcal{G}'_{t,d}$, that has for fixed $1 \le m \le n/d$, $dm$ out of $n$ vertices colored using colors $\{1,\ldots,dm\}$. Graphs in $\mathcal{G}'_{t,d}$ are obtained from graphs in $\mathcal{G}_{t,d}$ using the following procedure. The input to the  procedure is $G \in \mathcal{G}_{t,d}$ and a set $\{v_1,\ldots,v_{dm}\}$ of $dm$ vertices of $G$. For each $G \in \mathcal{G}_{t,d}$, we get a set of ${n \choose dm}(dm)!$ graphs of $\mathcal{G}'_{t,d}$ where each graph $G'$ is obtained by coloring the selected $dm$  vertices of $G$ by a permutation of $\{1,\ldots,dm\}$. A partially colored graph in $\mathcal{G}'_{t,d}$ is denoted $\langle H, U, g \rangle$ where $U \subset V(H), |U|=dm$, and $g: U \rightarrow \{1,\ldots,dm\}$.
\begin{proposition}
\label{prop:colorGprime1}
For each $t, d \geq 1$, $\log |\mathcal{G}'_{t,d}| \le \log |\mathcal{G}_{t,d}| + n \log n - (n-md)\log(n-md) - dm + O(\log n)$.
\end{proposition}
\begin{proof}
The $dm$ vertices given as input to the procedure can be selected in $n \choose dm$ ways and there are $(dm)!$ ways of coloring it. Thus, from each $G$ in $\mathcal{G}_{t,d}$ we get ${n \choose dm} (dm)!$ partially colored graphs. Hence, we have $|\mathcal{G}'_{t,d}| \le |\mathcal{G}_{t,d}| \frac{n!}{(n-dm)!}$. The inequality is due to the counting of partially colored graphs $\langle H_1, U_1, g_1 \rangle$ and $\langle H_2, U_2, g_2 \rangle$ such that $H_1$ and $H_2$ are isomorphic. Taking log on both sides and using Stirling's approximation, that is, $\log n!=n \log n - n + O(\log n)$, we get $\log |\mathcal{G}'_{t,d}| \le \log |\mathcal{G}_{t,d}| + n \log n - (n-md)\log(n-md) - dm + O(\log n)$. 
\end{proof}

\noindent
{\bf Graph Class $\mathcal{G}^c_{t,d}$.} As per the requirement of Proposition~\ref{prop:partialcoloring}, we construct a sub-class of $\mathcal{G}'_{t,d}$, denoted $\mathcal{G}^c_{t,d}$, for which we can obtain an exact count. We give a construction mechanism for graphs in $\mathcal{G}^c_{t,d} \subset \mathcal{G}'_{t,d}$ such that all graphs in $\mathcal{G}^c_{t,d}$ have the following properties:
\begin{itemize}
    \itemsep0em
    \item $dm$ vertices, denoted $U$, are fixed and have a fixed coloring,
    \item vertices in $U$ are colored using colors $\{1,\ldots,dm\}$, 
    \item $U$ induces a complete $d-$partite graph with $m$ vertices in each partition, and
    \item each partition $1 \le j \le d$ is colored using colors in the range $[(j-1)m+1, jm]$.
\end{itemize}

\noindent
In other words, for all the partially colored graphs in $\mathcal{G}^c_{t,d} \subset \mathcal{G}'_{t,d}$, $U$ and $g$ are fixed. The construction mechanism that constructs partially colored graph $\langle H, U, g \rangle$ is as follows. Based on the $t,d-$intersection representation, the vertices of a graph $H \in \mathcal{G}^c_{t,d}$ are of two types:
\begin{enumerate}
    \itemsep0em
    \item {\bf basis vertices $U$}: $dm$ vertices of $U$ are represented by a single $d-$dimensional box each, and
    \item {\bf dependent vertices $V(H) \backslash U$}: rest of the $(n-dm)$ vertices are represented by $t$ $d-$dimensional boxes.
\end{enumerate}

\noindent
The boxes corresponding to the basis and dependent vertices are called basis and dependent boxes, respectively. The input to the procedure that constructs $H$ are:
\begin{enumerate}
    \itemsep0em
    \item $n, m, d, t,$ and
    \item $\mathcal{J}=\{J_1,\ldots,J_{n-dm}\}$ where for $1 \le s \le n-dm, 1 \le j \le d, 1 \le p \le t$, $J_s=\{(e_{1,1},e'_{1,1}),\ldots,(e_{t,d},e'_{t,d})\}$ and $m(j-1)+\frac{m(p-1)}{t}+1 \le e_{p,j} \le e'_{p,j} \le m(j-1)+\frac{mp}{t}$.
\end{enumerate}
The following procedure constructs the $t,d-$intersection representation of $H$ and the interpretation of $\mathcal{J}_s$ will become clear after this. 
\begin{enumerate}
    \itemsep0em
    \item {\bf Construction of basis boxes.} Consider $dm$ pairwise disjoint unit length intervals called the \textit{basis intervals} such that each of the $d$ axes have $m$ basis intervals each. Each of these $m$ basis intervals on an axis $1 \le j \le d$ is denoted by $S_j$.
    %\begin{itemize}
    %    \itemsep0em
    %    \item $|S_j|=m$
    %    \item for $1\le j' \le d, I \in S_j$ and $I' \in S_{j'}$, $I \cap I'=\phi$.
    %\end{itemize}
    Also, let  $\texttt{l}(I)$ and $\texttt{r}(I)$ denote the left and right endpoints of interval $I$, respectively. For axis $j$ and $1 \le i \le m$, interval $I^i_j$ is defined by $\texttt{l}(I^i_j)=2i-1$ and $\texttt{r}(I^i_j)=2i$. This will ensure that the following hold: 
    \begin{itemize}
        \itemsep0em
        \item for $1 \le i \le m-1, \texttt{l}(I^i_j) < \texttt{l}(I^{i+1}_j)$,
        \item $\texttt{l}(I^{i+1}_j) = \texttt{r}(I^i_j)+1$, and
        \item $\texttt{l}(I^1_j)=1$ and $\texttt{r}(I^m_j)=2m$.
    \end{itemize}
    For $1 \le i \le m$, the interval $I^i_j$ is colored by $(j-1)m+i$. This ensures that all basis intervals are colored using colors from $1$ to $dm$. Construct $d-$dimensional basis boxes $b^i_j$ for each of the $dm$ vertices such that projection on all axes is the interval $[1,2m]$ except on axis $j$ where the projection is the interval $[2i-1,2i]$. In other words, for every $I^i_j$, extend it up to $2m$ along all of the axes $\{1,\ldots,d\}\backslash \{j\}$ to obtain a  $d-$dimensional box, $b^i_j$. Observe that extending $I^i_j$ along axis $j'$ will result in a rectangle and further extending it on axis $j''$ will give us a  3-dimensional box and so on. Every $b^i_j$ is assigned the same color  as $I^i_j$. Observe that the induced sub-graph corresponding to the vertices in $U$ forms a $d-$partite complete graph with each partition having $m$ vertices.

    \item {\bf Construction of dependent boxes.} Next, we consider $dt(n-dm)$ intervals called the \textit{dependent intervals}. Dependent intervals are constructed from basis intervals using input $\mathcal{J}$ in the following manner. The $m$ basis intervals on each axis is partitioned into $t$ blocks of $m/t$ intervals each. For every $J_s \in \mathcal{J}$, we obtain a set of $dt$ dependent intervals $\mathcal{I}_s=\{[\texttt{l}(I_{e_{1,1}}), \texttt{r}(I_{e'_{1,1}})],\ldots,[\texttt{l}(I_{e_{t,1}}), \texttt{r}(I_{e'_{t,1}})],$\\$ \ldots, [\texttt{l}(I_{e_{1,d}}), \texttt{r}(I_{e'_{1,d}})], \ldots, [\texttt{l}(I_{e_{t,d}}), \texttt{r}(I_{e'_{t,d}})]\}$  where $I_{e_{p,j}},I_{e'_{p,j}}$ are basis intervals of $S_j$ and $m(j-1)+\frac{m(p-1)}{t}+1 \le e_{p,j} \le e'_{p,j}\le m(j-1)+\frac{mp}{t}$ their colors. That is, the $p-$th interval on axis $j$ for a vertex is selected from the $p-$th block of $m/t$ basis intervals on that axis. The dependent intervals are left uncolored. There are $t$ boxes $\{b^1,\ldots,b^t\}$ created from these dependent intervals where $b^p=[l(I_{e_{p,1}}),r(I_{e_{p,1}})] \times [l(I_{e_{p,2}}),r(I_{e_{p,2}})] \times \ldots \times [l(I_{e_{p,d}}),r(I_{e_{p,d}})]$. 
    \item From the $t,d-$intersection representation we obtain the graph $H$ by establishing the following correspondence between set of basis and dependent boxes to $V(H)$. Basis boxes correspond to the colored $dm$ basis vertices in $U \subset V(H)$ and the set of $t$ dependent boxes $\{b_1,\ldots,b_t\}$ corresponds to an uncolored dependent vertex in $V(H) \backslash U$. Since basis boxes are fixed, $U$ is fixed and the coloring of basis boxes in the construction fixes $g$.
\end{enumerate}

\noindent
From the construction above we have the following lemma.
\begin{lemma}
    \label{lem:Gprimelsuperset}
    $\mathcal{G}^c_{t,d} \subseteq \mathcal{G}'_{t,d}$.
\end{lemma}
\begin{proof}
    From the construction given above, any graph $G \in \mathcal{G}^c_{t,d}$ has a $U \subset V(G)$ where $|U|=dm$ that is colored using colors $\{1,\ldots,dm\}$ such that:
    \begin{enumerate}
        \itemsep0em
        \item $G[U]$ is a complete $d-$partite graph, and 
        \item for $1 \le j \le d, 1 \le p \le t,$ a partition $j$ of $G[U]$ is colored using colors in the range $[(j-1)m+1,\ldots, jm]$.
    \end{enumerate}
    Since $\mathcal{G}^c_{t,d}$ is a special class of partially colored graphs with $d-$dimensional $t-$representation, $\mathcal{G}^c_{t,d} \subset \mathcal{G}'_{t,d}$. 
\end{proof}

\noindent
Figure~\ref{fig:lowerbound} shows an example of the construction of a graph $G \in \mathcal{G}^c_{2,2}$ in which the dependent vertex $v \in V(G)$ is represented by two boxes, namely, $b(v,1)$ and $b(v,2)$. $b(v,1)$ is constructed from dependent intervals colored $e_{1,1}, e'_{1,1}$ and $e_{2,1}, e'_{2,1}$. $b(v,2)$ is constructed from dependent intervals colored $e_{1,2}, e'_{1,2}$ and $e'_{2,2}$. $b(v,1)$ intersects the two-dimensional extensions of basis intervals colored $1,2,m+1,$ and $m+2$ whereas $b(v,2)$ intersects the two-dimensional extensions of basis intervals colored $m-1, m$ and $2m$. Thus, the neighbours of $v$ are vertices colored by $1, 2, m-1, m, m+1, m+2,$ and $2m$. 

\begin{figure}[ht]
\centering
\includegraphics[width=14.5cm,height=8cm]{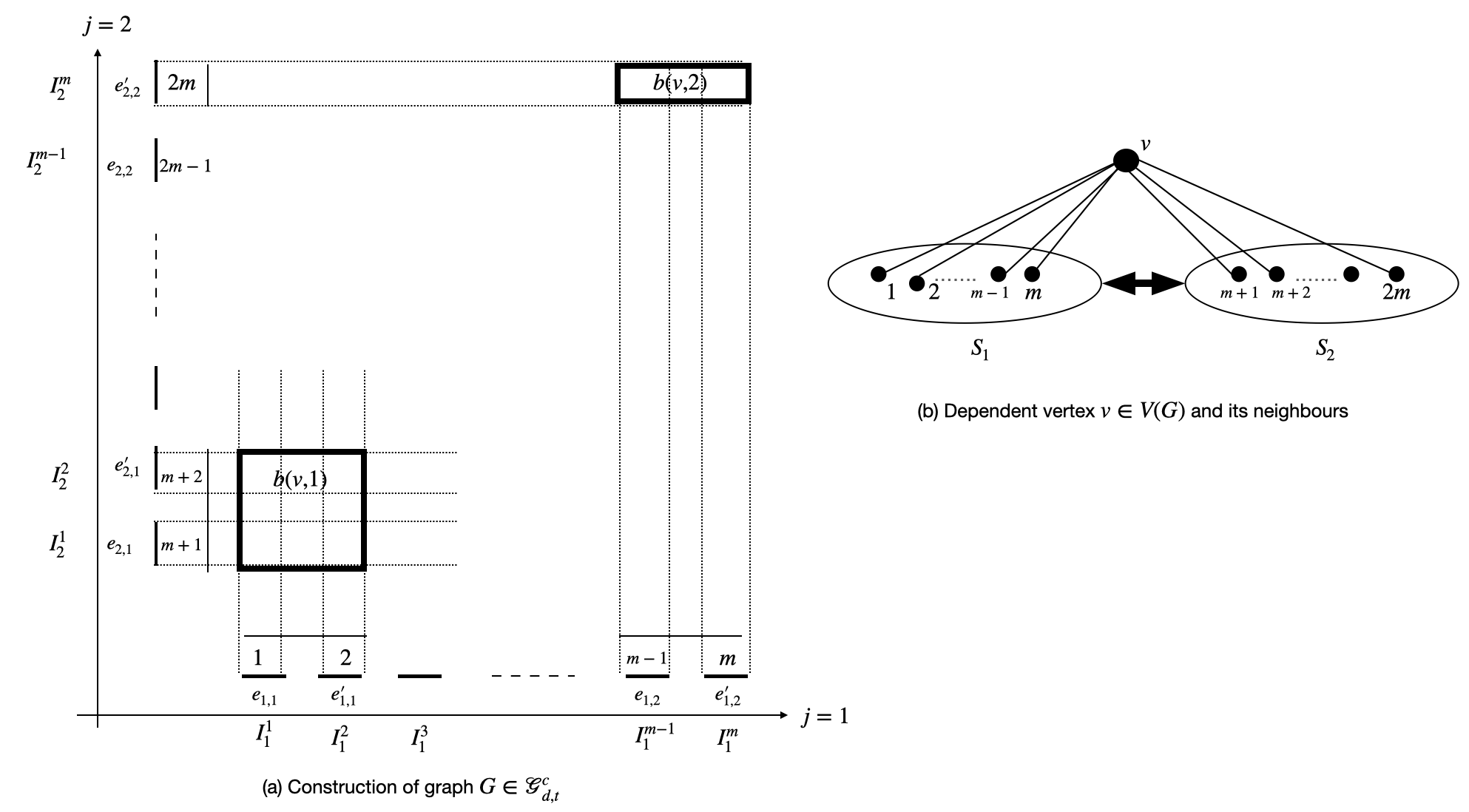}
\caption{(a) For $G \in \mathcal{G}^c_{2,2}$ and dependent vertex $v \in V(G)$ such that $J_v=\{(e_{1,1},e'_{1,1}),(e_{1,2},e'_{1,2}),(e_{2,1},e'_{2,1}),(e_{2,2},e'_{2,2})\}$ where $e_{1,1}=1,e'_{1,1}=2,e_{1,2}=m-1,e'_{1,2}=m,e_{2,1}=m+1,e'_{2,1}=m+2,e_{2,2}=2m,e'_{2,2}=2m$, (b) The colored neighbours of $v$ as per the basis intervals selected in (a). $S_1$ and $S_2$ contain $m$ vertices each and form an induced complete bipartite graph.}
\label{fig:lowerbound}
\end{figure}

\noindent
{\bf Computing $|\mathcal{G}^c_{t,d}|$}. In order to compute $|\mathcal{G}^c_{t,d}|$ and use Proposition~\ref{prop:partialcoloring}, we first note the following.

\begin{proposition}
    \label{lem:Gprimelb}
    $|\mathcal{G}'_{t,d}| \ge |\mathcal{G}^c_{t,d}|$.
\end{proposition}
\begin{proof}
    From Lemma~\ref{lem:Gprimelsuperset}. 
\end{proof}

\begin{comment}
that will be used  in Lemma~\ref{lem:nonisomorphic} to show that the construction procedure given above produces distinguishable graphs of $\mathcal{G}^c_{t,d}$. This is done in three steps as given below.
\begin{enumerate}
    \itemsep0em
    \item Observing that every $\mathcal{J}$ produces a graph of $\mathcal{G}^c_{t,d}$.
    \item In Proposition~\ref{prop:vertexdistinct} we observe that, for a fixed $p \in [t]$ and some $1 \le s \ne s' \le n-dm$, if $J_s \ne J_{s'}$ then the basis boxes intersecting the dependent box $b^s$ is distinct from those intersecting $b^{s'}$. This implies that the neighbourhood of vertices corresponding to $s$ and $s'$ are distinct because the colored neighbourhood of the vertices are different.
    \item For $1 \le p \le t$, let $D(b^p)$ denote the basis boxes intersecting dependent box $b^p$. In Corollary~\ref{cor:distinctnhb}, we observe that $\bigcup\limits_{1 \le p \le t} D(b^p)$ for $s$ and $s'$ are distinct.
\end{enumerate}
\end{comment}

\noindent
We have the following useful lemmas. The following lemma is a result of creating $t$ blocks of $m/t$ basis intervals in the construction procedure.
\begin{lemma}
\label{lem:distinctnhb}
For $1 \le p \ne p' \le t,$ let $D(b^p)$ denote the basis boxes intersecting dependent box $b^p$ of some dependent vertex $u$. Then there does not exist $p'$ such that $D(b^p) \cap D(b^{p'}) \ne \phi$, where $D^{p'}$ is another dependent box of $u$.
\end{lemma}
\begin{proof}
By definition, every dependent vertex is defined by $t$ boxes $\{b^1,\ldots,b^t\}$ created from dependent intervals where for $1 \le p \le t, b^p=[l(I_{e_{p,1}}),r(I_{e_{p,1}})] \times [l(I_{e_{p,2}}),r(I_{e_{p,2}})] \times \ldots \times [l(I_{e_{p,d}}),r(I_{e_{p,d}})]$. Note that for axis $j$ and $p'\ne p$, $e_{p',j}$ and $e'_{p',j}$ are colors that fall into the range $[m(j-1)+\frac{m(p'-1)}{t}+1,m(j-1)+\frac{mp'}{t}]$ which is not overlapping the range of colors for $e_{p,j}$ and $e'_{p,j}$. Thus, the basis boxes intersected by $D(b^p)$ have their projections on axis $j$ in the $p-$th block of $m/t$ basis intervals which is different from those intersecting $D(b^{p'})$ which have their projections on axis $j$ in the $p'-$th block of $m/t$ basis intervals. Thus, when $p\ne p'$, $D(b^p) \cap D(b^{p'}) \ne \phi$.
\end{proof} 

\noindent
The following lemma proves that different $\mathcal{J}$'s input to the construction process gives us graphs with uncolored dependent vertices having different colored vertices as neighbours. Let $K$ denote the set of all possible $\mathcal{J}$. 

\begin{lemma}
\label{lem:vertexdistinct}
    Let $\mathcal{J}, \mathcal{J}' \in K$ where $\mathcal{J}=\{J_1,\ldots,J_{n-dm}\}$ and $\mathcal{J}'=\{J'_1,\ldots,J'_{n-dm}\}$ such that for $1 \le s \le n-dm, J_s \ne J'_s$. Then the uncolored dependent vertex corresponding to $s$ has different colored basis vertices as neighbours in the graphs generated from $\mathcal{J}$ and $\mathcal{J}'$.
\end{lemma}
\begin{proof}
    For $1 \le p < p' \le t$, if a basis box $b$ can be selected as the colored neighbour of a dependent vertex by two different dependent boxes $b^p$ and $b^{p'}$ then $\mathcal{J}$ and $\mathcal{J}'$ both describe a single graph. In such a case, there are two ways of constructing the same graph in $\mathcal{G}^c_{t,d}$. However, as per Lemma~\ref{lem:distinctnhb}, by construction, a colored basis box can be selected as a neighbour by only one of the $t$ dependent boxes. From $J_s \ne J'_s$ we know that $\mathcal{I}_s \ne \mathcal{I}'_s$ for some dependent vertex $u$ corresponding to $s$. This implies that there exists $[l(I_{e_{p,j}}), r(I_{e'_{p,j}})]$ in $\mathcal{I}_s$ and $[l(I_{f_{p,j}}), r(I_{f'_{p,j}})]$ in $\mathcal{I}'_s$ such that $\{e_{p,j},e'_{p,j}\} \ne \{f_{p,j},f'_{p,j}\}$. So there exists a dependent box $b^p$ corresponding to $u$ for which the endpoints of intervals projected on axis $j$ are different. Thus, we conclude that $b^p$ intersects different colored basis boxes when $\mathcal{J} \ne \mathcal{J}'$ and so the uncolored dependent vertex corresponding to $b^p$ intersects different basis vertices. 
\end{proof}

%\begin{proposition}
%\label{prop:vertexdistinct}
%    Let $\mathcal{J}=\{J_1,\ldots,J_{n-dm}\}$ and $\mathcal{J}'=\{J'_1,\ldots,J'_{n-dm}\}$ be two distinct inputs to the construction procedure. For $1 \le s \le n-dm$, let $J_s \ne J'_s$. Then $\mathcal{I}_s \ne \mathcal{I}'_s$ for some dependent vertex corresponding to $s$ and there exists $[l(I_{e_{p,j}}), r(I_{e'_{p,j}})]$ in $\mathcal{I}_s$ and $[l(I_{f_{p,j}}), r(I_{f'_{p,j}})]$ in $\mathcal{I}'_s$ such that $\{e_{p,j},e'_{p,j}\} \ne \{f_{p,j},f'_{p,j}\}$.
%\end{proposition}

\noindent
The following is the central lemma used to obtain the lower bound. For $1 \le s \le n-dm$, let $\mathcal{J}(H)$ denote the $\mathcal{J} \in K$ that produces the graph $H \in \mathcal{G}^c_{t,d}$. Also, let $\mathcal{I}_s(H)$ denote a dependent vertex in $H$ defined by $\mathcal{I}_s$.

\begin{lemma}
\label{lem:distinct}
    Let $\langle H_1,U,g\rangle, \langle H_2,U,g\rangle$ be constructed from $\mathcal{J},\mathcal{J}' \in K$. Then $\langle H_1,U,g\rangle$ and $\langle H_2,U,g\rangle$ are same if and only if $\mathcal{J} = \mathcal{J}'$.
\end{lemma}
\begin{proof}
    We prove both directions as follows:
    \begin{enumerate}
        \item If $\langle H_1, U, g \rangle=\langle H_2, U, g \rangle$ then $\mathcal{J}=\mathcal{J}'$. We prove its contrapositive, that is, if $\mathcal{J} \ne \mathcal{J}'$ then $\langle H_1, U, g \rangle \ne \langle H_2, U, g \rangle$. If $\mathcal{J} \ne \mathcal{J}'$ then for $1 \le s \le n-dm$ there exists $J_s \ne J'_s$. Let $V$ denote the set of vertices in $H_1$ and $H_2$. By Lemma~\ref{lem:vertexdistinct}, there exists $u \in V\backslash U$ corresponding to $s$ such that $N_{H_1}(u) \cap U \ne N_{H_2}(u) \cap U$ where $N_{H_1}(u)$ and $N_{H_2}(u)$ are the neighbours of $u$ in $H_1$ and $H_2,$ respectively. Thus, $E(H_1) \ne E(H_2)$ and so by definition $\langle H_1, U, g \rangle \ne \langle H_2, U, g \rangle$.
        \item If $\mathcal{J} = \mathcal{J}'$ then $\langle H_1, U, g \rangle$ and $\langle H_2, U, g \rangle$ are same by construction.
    \end{enumerate}
\end{proof}
\begin{comment}
        \begin{enumerate}
        \itemsep0em
        \item If $\langle H_1,U_1,g_1\rangle$ and $\langle H_2,U_2,g_2\rangle$ are different then there are two cases to consider. In both cases, $U_1=U_2$ and $g_1=g_2$.
        \begin{enumerate}
            \itemsep0em
            \item $H_1$ and $H_2$ are non-isomorphic: For $1 \le s \le n-dm$, there exists a $u_s \in V(H_1)\backslash U_1$ and $v_s \in V(H_2) \backslash U_2$ described by $\mathcal{I}_s(H_1)$ and $\mathcal{I}_s(H_2)$, respectively, that have different neighbourhoods. This can happen only when $\mathcal{J}(H_1)$ is different from $\mathcal{J}(H_2)$. Thus, if $H_1$ is constructed from $\mathcal{J}$ then $H_2$ is constructed from $\mathcal{J}'$ such that $\mathcal{J} \ne \mathcal{J}'$.
            \item $H_1$ and $H_2$ are isomorphic: The construction procedure does not produce isomorphic graphs. Since $U_1=U_2$ and $g_1=g_2$ in the construction the $H_1$ and $H_2$ can be different only when one of the dependent vertex has a different neighbourhood. Thus, $H_1$ and $H_2$ cannot be isomorphic.
        \end{enumerate}
        \item If $\mathcal{J} \ne \mathcal{J}'$ then by Lemma~\ref{lem:vertexdistinct}, we know that if $\mathcal{J} \ne \mathcal{J}'$ then for some dependent vertex $u$, the $\mathcal{I}_s$ and $\mathcal{I}'_s$ are different and the colored basis boxes intersecting it are different. Thus, $\langle H_1,U_1,g_1\rangle$ and $\langle H_2,U_2,g_2\rangle$ are different.
    \end{enumerate}
\end{comment}
\noindent
In order to obtain $|\mathcal{G}^c_{t,d}|$ we prove the following lemma first. 

\begin{lemma}
    \label{lem:Gc_size}
    $|\mathcal{G}^c_{t,d}| = |K|$.
\end{lemma}
\begin{proof}
    The proof establishes a bijection between $\mathcal{G}^c_{t,d}$ and $K$ as follows. 
    \begin{itemize}
        \itemsep0em
        \item Every graph of $\mathcal{G}^c_{t,d}$ is produced by some element of $K$. Every $\mathcal{J} \in K$ produces a $G \in \mathcal{G}^c_{t,d}$. Since every set of intervals obtained from $\mathcal{J}$ describes some $t(n-dm)$ $d-$dimensional boxes that represent a graph this is true.
        \item For every $\mathcal{J} \in K$, a distinct graph of $\mathcal{G}^c_{t,d}$ is produced. As per Lemma~\ref{lem:distinct}, if $\mathcal{J}, \mathcal{J}' \in K$ are different then the graphs $H,H' \in \mathcal{G}^c_{t,d}$ produced from them are different.
    \end{itemize}
    Thus, $|\mathcal{G}^c_{t,d}| = |K|$.
\end{proof}

\noindent
We have the following important lemma.
\begin{lemma}
    \label{lem:Gclb}
    $\log |\mathcal{G}^c_{t,d}|=2dtn \log m - 2d^2 mt \log m -2dtn \log t + 2d^2tm \log t- dtn + d^2tm$.
\end{lemma}
\begin{proof}
    From Lemma~\ref{lem:Gc_size}, we know that $|\mathcal{G}^c_{t,d}|=|K|$. So we count number of $\mathcal{J}$ that can be obtained. Let $G \in \mathcal{G}^c_{t,d}$ and $T$ be the set  of $dm$ colored vertices. For a vertex $u \in V(G)\backslash T$ and $1\le j \le d$, there are $m/t$ ways of selecting $e_j$ or $e'_j$ from a block of $S_j$. The total number of ways $(e_j,e'_j)$ can be  selected from a block is $(m/t)^2$. But out of $(e_1,e'_1)$ and $(e'_1,e_1)$ we select only the one with first value not greater than the second. Since each pair can appear exactly twice, the total number of  ways of selecting $(e_j,e'_j)$ such that $e_j \le e'_j$ is $\frac{m^2}{2t^2}$. There  are $dt$ such pairs $(e_j,e'_j)$ that can be selected in $\big({\frac{m^2}{2t^2}}\big)^{dt}$ ways. Since each dependent vertex can be selected in $\big({\frac{m^2}{2t^2}}\big)^{dt}$ ways, the total number of ways of  selecting $(n-dm)$ dependent vertices is ${\big(\frac{m^2}{2t^2}}\big)^{dt(n-dm)}$. Since, $|\mathcal{G}^c_{t,d}|=|K|$ we have $|\mathcal{G}^c_{t,d}|={\big(\frac{m^2}{2t^2}}\big)^{dt(n-dm)}$. Thus, $\log |\mathcal{G}^c_{t,d}|=2dtn \log m - 2d^2 mt \log m -2dtn \log t + 2d^2tm \log t- dtn + d^2tm$. Since $\mathcal{G}^c_{t,d} \subseteq \mathcal{G}'_{t,d}$, we have $\log |\mathcal{G}'_{t,d}| \ge \log |\mathcal{G}^c_{t,d}|$.    
\end{proof}

\noindent
{\bf Computing $|\mathcal{G}_{t,d}|$.} The following theorem gives the lower bound for $|\mathcal{G}_{t,d}|$ using Proposition~\ref{prop:partialcoloring}, Proposition~\ref{prop:colorGprime1}, and Lemma~\ref{lem:Gclb}.
\lowerbound*
\begin{proof}
Using Proposition~\ref{prop:partialcoloring} and substituting the expression for $\log |\mathcal{G}'_{t,d}|$ from Proposition~\ref{prop:colorGprime1} and expression for $\log |\mathcal{G}^c_{t,d}|$ from Lemma~\ref{lem:Gclb}, we get the following.
\begin{align*}
\log|\mathcal{G}_{t,d}| & \ge  2dtn \log m - n \log n -2d^2mt \log m -dtn \log t -dtn -O(\log n)
\end{align*}
To simplify the above expression, we substitute $m=\frac{n}{td^2\log n}$. The construction of a graph in $\mathcal{G}^c_{t,d}$ is well-defined only when $m \ge 1$ and for $m=\frac{n}{td^2\log n}$, this happens only when $td^2 \le \frac{n}{\log n}$.  Further, when $td^2$ is $o(n / \log n)$, it follows that $\log|\mathcal{G}_{t,d}|  \ge  (2dt-1) n \log n - 4dtn \log d - 4dtn \log t - 2dtn \log\log n-2tn - dtn-O(\log n)$.
\end{proof}

%\begin{corollary}
%\label{cor:lbtintdbox}  
%$\log |\mathcal{G}_t|$ is lower bounded by $(2t-1)n \log n - 4tn \log t - 2tn \log \log n-2n-tn-O(\log n)$ and $\log |\mathcal{G}_d|$ by $(2d-1)n \log n - 4dn \log d - 2dn \log \log n-2n-dn-O(\log n)$.
%\end{corollary}
\lbcorollary*
\begin{proof}
Obtained directly by substituting $d=1$ and $t=1$ in the  expression for $\mathcal{G}_{t,d}$ given in Theorem~\ref{thm:succreplbkinterval}, respectively.    
\end{proof}

\noindent
For each $t, d \geq 1$, asymptotically the largest subtracted term in the lower bound expression for  $\log|\mathcal{G}_{t,d}|$  is  $2dtn \log\log n$.  We make the following observation which is useful in the analysis of the data structure that is presented in the next section.
\begin{observation}
\label{obs:range}
For $t,d \ge 1, td^2$ in $ o( n/\log n)$, and sufficiently large $n$, $\log|\mathcal{G}_{t,d}| \ge dtn \log n$.
\end{observation}

\begin{comment}
\begin{lemma}
\label{lem:lbkinterval}
$\log|\mathcal{G}_t|  \ge  (2t-1)n \log n - 4tn \log t - 2tn \log \log n-2n-tn-O(\log n)$.    
\end{lemma}
\begin{proof}
Directly obtained by substituting $d=1$ in $\mathcal{G}_{t,d}$ in Lemma~\ref{lem:succreplbkinterval}.    
\end{proof}

\begin{lemma}
\label{lem:lbkboxicity}
$\log|\mathcal{G}_d|  \ge  (2d-1)n \log n - 4dn \log d - 2dn \log \log n-2n-dn-O(\log n)$.    
\end{lemma}
\begin{proof}
 Directly obtained by substituting $t=1$ in $\mathcal{G}_{t,d}$ in Lemma~\ref{lem:succreplbkinterval}.   
\end{proof}
\end{comment}

\section{Succinct Representation of Graphs with $d-$Dimensional $t-$Representation}
\label{sec:succrep}
In this section, we present a data structure for unlabeled graphs $G = (V, E)$ with $d-$dimensional $t-$representation. Our data structure occupies a space of ($(2dt-1)n \log n + 2dtn \log t + o(dtn\log n)$) bits, which is succinct when $td^2$ is in $o(n/\log n)$ from the result of Section~\ref{sec:kintlowerbound}. Our representation's space usage matches the upper bound given by Erd\H{o}s and West~\cite{EW} (note that their work considers the labeled case), while also efficiently supporting $\adj{}$ and $\neighbor{}$ queries. Note that $\degr{}$ queries can be answered either by responding to the $\neighbor{}$ query (thus sharing the same time complexity) or by explicitly storing the answers using an additional $n \log n$ bits (thus supporting $\degr{}$ query in $O(1)$ time). However, in the latter case, the data structure is not succinct when both $t$ and $d$ are $O(1)$.

Suppose $V = [n]$, and the $t,d-$intersection representation $\mathcal{I} = \{I^v \mid v \in V\}$ of $G$ is provided as an input. According to the definition  of $t,d-$intersection representation, each $I^v$ is represented as a set of $t$ disjoint $d$-dimensional intervals $\{I^v_1, \dots, I^{v}_t\}$, where for $p \in [t], I^v_p = [l^v_{p, 1}, r^v_{p, 1}] \times [l^v_{p, 2}, r^v_{p, 2}] \times \dots \times [l^v_{p, d}, r^v_{p, d}]$. Note that if $I^v$ originally contains $t' < t$ $d$-dimensional intervals, we add $t-t'$ dummy $d$-dimensional intervals to $I^v$. These dummy intervals are designed not to intersect with any other $d$-dimensional intervals in $\mathcal{I}$ (including other dummy intervals). Without loss of generality, for any $u, v \in V$, $p, q \in [t]$, we assume the following: (i) $l^u_{1, 1} < l^v_{1, 1}$ if $u < v$, and (ii) $r^u_{p, 1} < l^u_{q, 1}$ if $p < q$. Next, for each $j \in [d]$, let $I(j) = \{I^v_{p, j} \mid v \in V, p \in [t]\}$ be a set of $tn$ intervals that collects all intervals of the $j$-th dimension in $\mathcal{I}$. We assign a distinct integer from $1$ to $2tn$ to each endpoint in $I(j)$~\cite{Hanlon}.

\subsection{Some Useful Data Structures} 
Our succinct data structure is constructed using the following succinct data structures.

\noindent\textbf{Rank and Select Queries.} Given a string $S$ of size $n$ over an alphabet $\Sigma = \{0, 1, \dots, \sigma\}$, $\rank$ and $\select$ queries can be defined on $S$ for any $\alpha \in \Sigma$ as follows: 
(i) $\rank_{\alpha}(i, S)$: returns the number of $\alpha$ in $S[1, \dots, i]$, and (ii) $\select_{\alpha}(i, S)$: returns the position of $i$-th $\alpha$ in $S$. The following lemma by Golynski et al. in~\cite{DBLP:conf/soda/GolynskiMR06} shows that it is possible to store $S$ in a succinct space while efficiently supporting both $\rank$ and $\select$ queries.

\begin{lemma}[\cite{DBLP:conf/soda/GolynskiMR06}]\label{lem:rankselect}
   Given a string $S$ of size $n$ over an alphabet $\Sigma = \{0, 1, \dots, \sigma\}$, there exists an $(n \log (\sigma+1) + o(n \log \sigma))$-bit representation of $S$ that can answer $\rank$ queries in $O(\log \log (1 + \sigma))$ time and $\select$ queries in $O(1)$ time. Additionally, the representation allows access to any position in $S$ in $O(\log \log (1 + \sigma))$ time.
\end{lemma}

\noindent{\bf Permutations.} The following  data structure by  Munro et al.~\cite{DBLP:journals/tcs/MunroRRR12}, gives a  succinct representation for storing permutation of $[n]$.
\begin{lemma}[\cite{DBLP:journals/tcs/MunroRRR12}]\label{lem:perm}
Given a permutation of $[n]$ there exists an $(n\log n+ o(n\log n))$-bit data structure that supports the following queries.
\begin{itemize}
    \itemsep0em
    \item $\pi(i)$: Returns the $i-$th value in the permutation in $O(1)$ time.
    \item $\pi^{-1}(j)$: Returns the position of the $j-$th value in the permutation in $O(f(n))$ time for any increasing function $f(n)=o(\log n)$.
\end{itemize}
\end{lemma}

\subsection{Succinct Representation} 
For each $j \in \{1, \dots, d\}$, we store a sequence $S_{j}[1, 2tn]$ of size $2tn$ over an alphabet $[2t]$ defined as follows. For $p \in [t]$, $S_{j}[i]$ is assigned the value $2(p-1)$ if the point $i$ corresponds to the left endpoint of an interval in $\{[l^v_{p, j}, r^v_{p, j}] \mid v \in V\}$, and $2p-1$ if it corresponds to the right endpoint. Each $S_{j}$ is stored using the data structure of Lemma~\ref{lem:rankselect}, allowing for efficient $\rank{}$ and $\select{}$ queries in $O(\log \log t)$ and $O(1)$ time, respectively. The total space required to store all $S_{j}$ sequences is $2dtn \log t + d \cdot o(tn \log t)$ bits.

Additionally, we maintain $(2dt-1)$ permutations $\{\pi_{(p, j)} \mid (p, j) \in [t] \times [d] \setminus (1, 1)\}$ and $\{\rho_{(p, j)} \mid (p, j) \in [t] \times [d]\}$ on $[n]$. These permutations store the ranks of $l^v_{p, j}$ and $r^v_{p, j}$ in $S_{j}$, respectively. Specifically, $\pi_{(p, j)}(v)$ represents $\rank{}_{(2(p-1))}(l^v_{p, j}, S_{j})$, and $\rho_{(p, j)}(v)$ represents $\rank{}_{(2p-1)}(r^v_{p, j}, S_{j})$. Each permutation is stored using the data structure of Lemma~\ref{lem:perm}, which supports $\pi_{(p, j)}$ and $\rho_{(p, j)}$ in $O(1)$ time, and $\pi_{(p, j)}^{-1}$ and $\rho_{(p, j)}^{-1}$ in $O(f(n))$ time for any increasing function $f(n) = o(\log n)$. The total space required for storing these permutations is $(2dt-1)n \log n + o(dtn\log n)$ bits. 

\begin{figure}
	\begin{center}
		\includegraphics[scale=0.26]{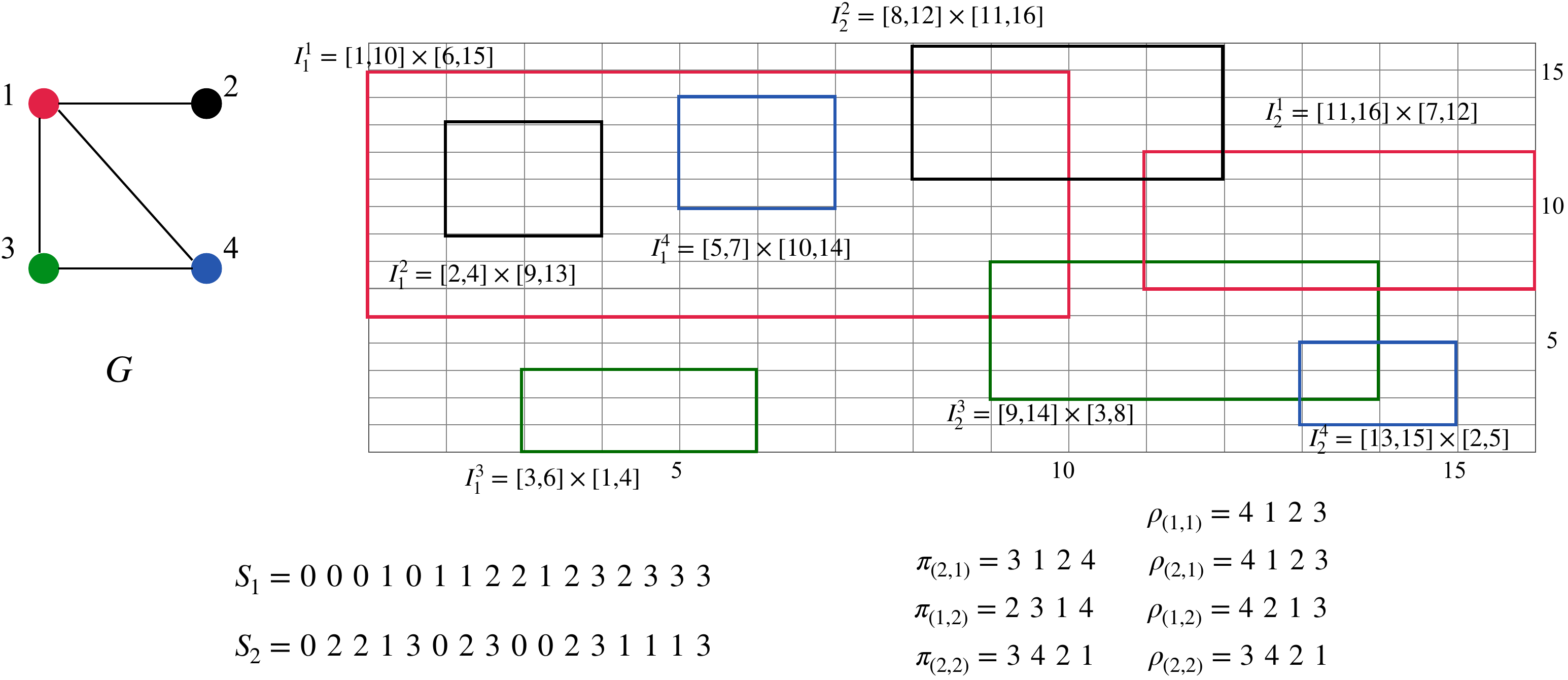}
	\end{center}
 %\captionsetup{font=small} 
	\caption{Example of the data structure on $G$, represented as $2, 2$-intersection representation. Note that there also exists a $1, 1$-intersection representation of $G$. }
	\label{fig:ub_example}
\end{figure}

Overall, the space usage of the data structure is $(2dt-1)n \log n + dtn \log t + o(dtn\log n)$ bits
%({\color{red}{$2dtn \log t$?}}) 
(see Figure~\ref{fig:ub_example} for an example of the data structure). Using this data structure, for any vertex $v \in V$, we can decode the interval $[l^v_{p, j}, r^v_{p, j}]$ in $O(1)$ time by:
(i) computing $l^v_{1, 1}$ using $\select_{0}(v, S_1)$, and
(ii) computing $l^v_{p, j}$ (if necessary) and  $r^v_{p, j}$ using $\select_{(2(p-1))}(\pi_{(p, j)}(v), S_{j})$ and $\select_{(2p-1)}(\rho_{(p, j)}(v), S_{j})$, respectively 
%({\color{red}{$\select_{(2(p-1))}(\pi_{(p, j)}(v), S_{j})$ and $\select_{(2p-1)}(\rho_{(p, j)}(v), S_{j})$}}).
In the following, we will describe how to support $\adj{}$ and $\neighbor{}$ queries on graph $G$ using this data structure.
\\\\
\noindent\textbf{Answering $\adj{}(u, v)$ Query. } To determine whether two vertices $u$ and $v$ are adjacent in the graph $G$, we need to verify if there exist $p$ and $q$ such that for all $j \in [d]$, $I^u_{p, j} \cap I^v_{q, j} \neq \emptyset$. To accomplish this efficiently, we can employ a line sweep procedure on the first dimension of the intervals in $\mathcal{I}$ as follows:

\begin{enumerate}
    \item Initialize $p$ and $q$ as $1$.
    \item Check if $I^u_{p,1}$ and $I^v_{q,1}$ intersect. If so, check $I^u_{p, j} \cap I^v_{q, j} \neq \emptyset$ for all $j \in \{2, \dots, d\}$.
    \item If $I^u_{p,1} \cap I^v_{q,1} = \emptyset$, increase $p$ by $1$ if $l^u_{p, 1} < l^v_{q,1}$, or increase $q$ by $1$ if $l^u_{p, 1} > l^v_{q,1}$. Repeat step 2 until $p \leq t$ and $q \leq t$.
\end{enumerate}

Note that if $I^u_{p,1}$ and $I^v_{q,1}$ do not intersect and $l^u_{p, 1} < l^v_{q,1}$, $l^v_{q,1}$ can only intersect intervals in $\{I^u_{p+1, 1}, \dots, I^u_{t, 1}\}$, as $l^u_{p, 1} < l^u_{p', 1}$ for any $p < p'$ (the case when $l^u_{p, 1} > l^v_{q,1}$ is analogous). This means we only need to check the intersection of at most $2t$ pairs of intervals. Additionally, we can check whether $I^u_{p,d}$ and $I^v_{q,d}$ intersect in $O(1)$ time by computing the four endpoints of the intervals. The second step of the procedure takes $O(d)$ time. Since the third step in the worst case iterates $O(t^2)$ times we can answer the $\adj{}(u, v)$ query in $O(dt^2)$ time.
\\\\
\noindent\textbf{Answering $\neighbor{}(u)$ Query. } 
Indeed, we can straightforwardly answer the $\neighbor{}(u)$ query in $O(dtn)$ time by checking $\adj{}(u, v)$ for all vertices $v \in V$. 
Now, we describe an alternative algorithm for $\neighbor{}$ queries when $d = 1$, using an additional $O(tn)$ bits of space. 
In this specific case, we  utilize the following lemma, which is derived from the work of Acan et al.~\cite{HSSS} for interval graphs.
Also in this case, we use $I^v_p$ to denote $I^v_{p,1}$.

\begin{lemma}[\cite{HSSS}]\label{lem:neighbor}
% Suppose there exists a representation of an interval graph $G = (V, E)$ with $n$ vertices that allows decoding of the endpoints of $I^v$ in $O(1)$ time for any $v \in V$.  Then after counting the number of left endpoints before the right endpoint of $I^v$ 
% %in $O(r)$ time 
% , it is possible to report the rank of right endpoints of all intervals intersecting with $I^v$ in $O(1)$ time per interval using $O(n)$ bits of additional space.
Suppose there exists a set of $n$ intervals whose endpoints are all distinct integers from $U$.
Then for any interval $I^v = [s, t]$ where $s, t \in U$, there exists an $O(n)$-bit data structure to report the rank of right endpoints of all intervals in the set intersecting with $I^v$ in $O(1)$ time per interval, provided the number of left endpoints preceding the right endpoint of $I^v$ is known.
\end{lemma}

For each $p \in [t]$, we construct the data structure described in Lemma~\ref{lem:neighbor} on the intervals in $\{I^v_p \mid v \in V\}$, requiring a total of $O(tn)$ additional bits. Now we report all the vertices $v$ where $I^u_p$ intersects with any interval in $\{I^v_1, \dots, I^v_t\}$ using the following procedure:

\begin{enumerate}
    \item Compute the interval $I^u_p = [l^u_p, r^u_p]$ in $O(1)$ time.
    \item For each $q \in [t]$, perform the following steps:
    \begin{enumerate}
        \item Count the number of left endpoints in $\{l^v_q \mid v \in V\}$ that appear before $r^u_p$ in $O(\log \log t)$ time using $\rank{}_{2(q-1)}(r^u_p, S_{1})$. 
        %({\color{red}{Where is the count used?}})
        \item Utilize the data structure from Lemma~\ref{lem:neighbor} on $\{I^v_q \mid v \in V\}$ to report the rank of the right endpoint in the set 
        $Q = \{I^v_q \mid I^v_q \cap I^u_p \neq \emptyset\}$. Then for each interval $I^v_q \in Q$ whose rank of the right endpoints is $k$ among $\{e^1_q \dots, e^n_q\}$, report the corresponding vertex $v$ in $O(f(n) + \log \log t)$ time by returning $\rho_{(q, 1)}^{-1}(k)$.
        Thus, it is possible to report all vertices in the set $\{v \mid I^v_q \cap I^u_p \neq \emptyset\}$ within $O((f(n) +\log \log t) \cdot \degr{}(u))$ time.
    \end{enumerate}
\end{enumerate}
By repeating the second step of the above procedure for all $q \in [t]$, report all the vertices $v$ where $I^u_p$ intersects with any interval in ${I^v_1, \dots, I^v_t}$ in $O(t (f(n) + \log \log t) \cdot \degr{}(u))$ time.
Thus, we can answer the $\neighbor{}(u)$ query in $O(t^2 (f(n) + \log \log t) \cdot \degr{}(u))$ time by performing the above procedure for all $p \in [t]$.
We summarized the results in the following theorem.

\datastructure*

\succinct*
\begin{proof}
For a data structure for $\mathcal{G}_{t,d}$ to be succinct it must take at most $\log |\mathcal{G}_{t,d}| + o(\log |\mathcal{G}_{t,d}|)$ bits of space. We will show that this is true for the data structure of Theorem~\ref{thm:tdintervalub}. For $t,d \ge 1$ and $td^2$ in $o( n/\log n)$, we have,
$$
    \log|\mathcal{G}_{t,d}| \ge (2dt-1)n \log n - 4dtn \log d - 2dtn \log t - 2dtn \log \log n - 2n -dtn - O(\log n)
$$
Rearranging the terms we have,
\begin{equation}
\label{eq:eq1}
\begin{split}
    (2dt-1)n \log n \le  \log|\mathcal{G}_{t,d}| + 4dtn \log d + 2dtn \log t + 2dtn \log \log n + 2n + dtn + O(\log n)
\end{split}
\end{equation}
Substitute Equation~\ref{eq:eq1} in $((2dt-1)n \log n + 2dtn\log t + o(dtn\log n))$. Observing that apart from $\log |\mathcal{G}_{t,d}|$, $2dtn \log\log n$ is the largest term and $2dtn \log\log n$ is in $o(dtn \log n)$, we have,
\begin{equation}
\label{eq:eq2}
\begin{split}
    (2dt-1)n \log n + 2dtn\log t + o(dtn\log n) \le  \log|\mathcal{G}_{t,d}| + o(dtn\log n)
\end{split}
\end{equation}
From Observation~\ref{obs:range}, we know that $\log|\mathcal{G}_{t,d}| \ge dtn\log n$. This means, for any $f$ in $o(dtn \log n)$, $f$ is also in $o(\log |\mathcal{G}_{t,d}|)$. Thus, from Equation~\ref{eq:eq2} we get the following.
$$
(2dt-1)n \log n + 2dtn\log t + o(dtn\log n) \le  \log|\mathcal{G}_{t,d}| + o(\log|\mathcal{G}_{t,d}|)
$$
Hence, the data structure of Theorem~\ref{thm:tdintervalub} that takes $(2dt-1) n \log n + 2dtn \log t + o(dtn\log n)$ is succinct.
\end{proof}

\begin{remark}
By maintaining the data structure proposed by Lee and Wong~\cite{DBLP:journals/jal/LeeW81} and utilizing an additional $O(tn \log^{2d} n)$ bits, we can handle $tn$ $d$-dimensional intervals $\{\prod_{j=1}^{d}I^v_{p, j} \mid v \in V, p \in [t]\}$ efficiently. With this data structure, $\neighbor{}(u)$ queries can be answered in $O(t^2 \cdot (\log^{2d-1} n + \degr(u) \cdot \log \log t))$ time. When compared to the data structure described in Theorem~\ref{thm:tdintervalub}, this approach employs succinct space when $t^2 = o(n / \log n)$ and $d = o(\log\log n / \log \log \log n)$, while still achieving faster query times when $t = o(d)$ and $\deg{}(u) = o(n/\log \log t)$.
\end{remark}

By combining the data structure of Theorem~\ref{thm:tdintervalub} with the maximum boxicity of graphs with bounded number of edges~\cite{DBLP:journals/ejc/Esperet16}, and the maximum interval number of graphs with bounded degree~\cite{GW}, we can derive the following corollaries:

\begin{corollary}
Given a $G= (V, E)$  with $n$ vertices and $m$ edges, the following results hold:
\begin{itemize}
    \item There exists an $O(n\sqrt{m \log m} \log n)$-bit data structure on $G$ that can answer $\adj{}$ queries in $O(\sqrt{m \log m})$ time.
    \item If the maximum degree of $G$ is $\Delta$, there exists an $(n\Delta  \log n + o(n\Delta \log n))$-bit data structure on $G$ that can answer $\adj{}$ queries in $O(\Delta)$ time, and $\neighbor{}$ queries in $O(\Delta^2 (f(n) + \log \log \Delta) \cdot \degr{}(u))$ time. Here, $f(n)$ is any increasing function in $o(\log n)$.
\end{itemize}
\end{corollary}

\subsection{Hardness of $\neighbor{}$ Query on Graphs with Large Interval Numbers }
When $d = 1$ and $t = \Theta(n)$, the data structure described in Theorem~\ref{thm:tdintervalub} can answer $\neighbor{}$ queries in $O(n^2)$ time. In this section, we show that if $G$ is provided with $t$-interval representation whose interval number of $G$ is $\Theta(n)$, the data structure described in Theorem~\ref{thm:tdintervalub} achieves an asymptotically optimal query time for both $\adj{}$ and $\neighbor{}$ queries, unless the combinatorial Boolean matrix multiplication (BMM) conjecture~\cite{DBLP:conf/stoc/HenzingerKNS15} fails (in this section we only consider combinatorial algorithms, i.e., using no algebric properties of the structures), assuming that $G$ is given as a $t, 1$-intersection representation.
More precisely, given a $t, 1$-intersection representation of $G$, we show that if there exists a data structure with construction time $c(t, n)$ that can answer $\neighbor{}(v)$ queries in $\alpha(t, n)$ time, it is possible to solve $BC$ in the Boolean semi-ring using an $O(n\alpha(t, n) + c(t, n) + nt)$-time algorithm, where $B$ and $C$ are Boolean matrix of size $n \times t$ and $t \times n$, respectively. This result implies that when $t = \Theta(n)$, for any constants $0 < \delta, \epsilon < 1$, it is impossible to answer $\neighbor{}$ queries in an amortized $O(n^{2-\delta} \cdot \text{polylog}(n))$ time per neighbor using any data structure occupying $O(n^{3-\epsilon})$ bits of space, unless the BMM conjecture is false.

\begin{figure}
	\begin{center}
		\includegraphics[scale=0.3]{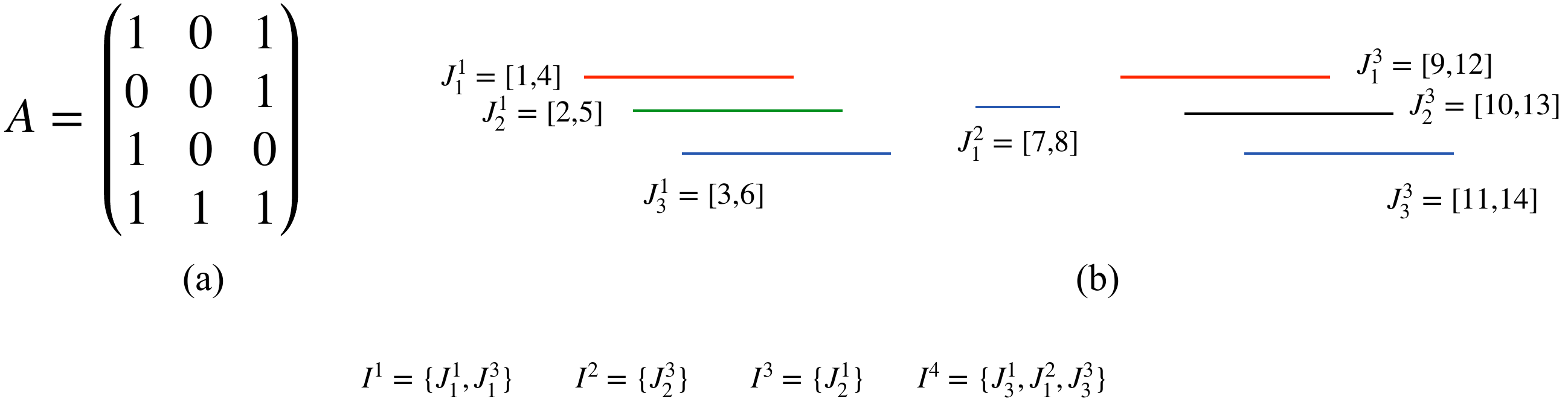}
	\end{center}
 %\captionsetup{font=small} 
	\caption{$4 \times 3$ Boolean matrix $A$ and the $3,1-$intersection representation of $G_A$.}
	\label{fig:hardness}
\end{figure}

\hardness*
\begin{proof}
    To show the theorem, we can focus on computing $AA^T$ for any Boolean matrix $A$ of size $n \times t$ that can be performed in $O(n\alpha(t, n) + c(t, n))$ time. Note that the multiplication of any two arbitrary $n/2 \times t$ matrix $B$ and $t \times n/2$ matrix $C$ 
    %({\color{red}{Isn't the size $n/2 \times t$ and $t \times n/2$?}}) 
    can be computed using $AA^T$, where $A$ is constructed as $A = \begin{pmatrix}  B \\ C^T \\ \end{pmatrix}$. Now we construct an interval representation of a graph $G_A$ with $n$ vertices and an interval number of $t$ using the following procedure, which can be accomplished in $O(nt)$ time (see Figure~\ref{fig:hardness} for an example):
    
    \begin{enumerate}
        \item Initialize variables $s$ and $p$ to $1$.
        \item Let $i_p$ denote the number of $1$s in the $p$-th column. Construct $i_p$ intervals $J^p_{1}, J^p_{2}, \ldots, J^p_{i_p}$ as follows: For each $\ell \in [i_p]$, set $J^p_{\ell} = [s+(\ell-1), s+i_p+(\ell-1)]$.
        \item Increment $s$ by $2i_p$, and increment $p$ by $1$. Repeat step $2$ until $p$ becomes $t$.
        \item After constructing the intervals, for each $1$ entry $(v, p)$ of matrix $A$, let $v_p$ be the rank of the $1$ entry in column $p$ of $A$, denoted as $A_p$. Define $I^v$ as the set of intervals $\{J^p_{v_p} \mid (v, p)\text{~is a $1$ entry of~}A\}$. Since $|I^v|$ is at most $t$, the graph $G_A$ has a $t,1-$intersection representation.
    \end{enumerate}

    Once $G_A$ is constructed using $c(t, n)$ time, we proceed to build the data structure that can answer $\neighbor{}$ queries in $\alpha(t, n)$ time. Then it is clear that for any vertex $u \in V$, $AA^{T}_u$ can be computed in $O(\alpha(t, n) + n)$ time, by setting the $v$-th row of $AA^{T}_u$ to $1$ if and only if $u = v$ or $v$ is a neighbor of $u$. Consequently, the entire matrix $AA^T$ can be computed in $O(n\alpha(t, n) + nt + c(t, n))$ time.
\end{proof}

% Note that within the proof of Theorem~\ref{thm:hardness}, each entry of $AA^{T}$ can also be computed using a single $\adj{}$ query. As a result, Theorem~\ref{thm:hardness} implies that when $t=\Theta(n)$, it is impossible to answer $\adj{}$ queries in an amortized $O(n^{1-\delta} \cdot \text{polylog}(n))$ time using any data structure that occupies $O(n^{3-\epsilon})$ bits of space, unless the BMM conjecture fails ({\color{red}{What about adjacency matrix?}}). 
Note that the total construction time of the data structure presented in Theorem~\ref{thm:tdintervalub} is $O(tn \times \text{polylog}(n))$~\cite{DBLP:journals/tcs/MunroRRR12,DBLP:conf/soda/GolynskiMR06}. 
Thus, Theorem~\ref{thm:hardness} implies that assuming the validity of the BMM conjecture, the data structure presented in Theorem~\ref{thm:tdintervalub} offers an asymptotically optimal query time (within a polylogarithmic factor) for both $\adj{}$ (in amortized) and $\neighbor{}$ queries when $d = 1$ and $t = \Theta(n)$, when the $t,d-$intersection representation of $G$ is provided.

\section{Conclusion}
\label{sec:conclusion}
As potential future works, we think that the following improvements to our results need to be addressed. In Theorem~\ref{thm:succreplbkinterval}, we have given a conditional lower bound for $\mathcal{G}_{t,d}$, and it is desirable to improve this by giving an unconditional lower bound. The possibility of a different data structure for $\mathcal{G}_{t,d}$ that gives better adjacency query time needs to be investigated along with a lower bound proof on the query times conditioned on space usage. Bukh and Jeffs~\cite{bukh2022enumeration} recently gave a tight counting lower bound for $d$-representable simplicial complexes which are another generalization of interval graphs, and it remains to be examined whether our techniques here could lead to the succinct data structure for such combinatorial objects.

\Section{References}
\bibliographystyle{IEEEbib}
\bibliography{refs}

\end{document}